\documentclass{article}
\usepackage{amsmath,amssymb,amsthm,authblk}
\newtheorem{definition}{Definition}
\newtheorem{theorem}{Theorem}
\newtheorem{lemma}{Lemma}
\newtheorem{observation}{Observation}

\newtheorem{conjecture}{Conjecture}
\newcommand{\tbd}{support}
\usepackage{todonotes}
\newcommand{\degleq}[1]{\ensuremath V^{{\scriptscriptstyle \leq} #1}}
\newcommand{\deggt}[1]{\ensuremath V^{{\scriptscriptstyle >}#1}}
\newcommand{\nbrsleq}[2]{\ensuremath N_{#1}^{{\scriptscriptstyle \leq}#2}}
\newcommand{\nbrsgt}[2]{\ensuremath N_{#1}^{{\scriptscriptstyle >}#2}}
\newcommand{\low}{\degleq{3}}
\newcommand{\high}{\deggt{3}}
\newcounter{claim}
\newenvironment{claim}
{\medskip\par\noindent\itshape\refstepcounter{claim}Claim~\theclaim.}{\medskip\par}
\begin{document}
\title{Adjacent vertex distinguishing total coloring of 3-degenerate graphs}
\author[1]{Diptimaya Behera}
\author[1]{Mathew C. Francis}
\author[2]{Sreejith K. Pallathumadam}
\affil[1]{Indian Statistical Institute, Chennai Centre}
\affil[2]{Indian Institute of Technology Madras}
\date{}
\maketitle
\begin{abstract}
A total coloring of a simple undirected graph $G$ is an assignment of colors to its vertices and edges such that the colors given to the vertices form a proper vertex coloring, the colors given to the edges form a proper edge coloring, and the color of every edge is different from that of its two endpoints. That is, $\phi:V(G)\cup E(G)\rightarrow\mathbb{N}$ is a total coloring of $G$ if $\phi(u)\neq\phi(v)$ and $\phi(uv)\neq\phi(u)$ for all $uv\in E(G)$, and $\phi(uv)\neq\phi(uw)$ for any $u \in V(G)$ and distinct $v,w \in N(u)$ (here, $N(u)$ denotes the set of neighbours of $u$).
A total coloring $\phi$ of a graph $G$ is said to be ``Adjacent Vertex Distinguishing'' (or AVD for short) if for all $uv\in E(G)$, we have that $\phi(\{u\}\cup\{uw:w\in N(u)\})\neq\phi(\{v\}\cup\{vw\colon w\in N(v)\})$.
The AVD Total Coloring Conjecture of Zhang, Chen, Li, Yao, Lu, and Wang (\textit{Science in China Series A: Mathematics}, 48(3):289--299,
2005) states that every graph $G$ has an AVD total coloring using at most $\Delta(G)+3$ colors, where $\Delta(G)$ denotes the maximum degree of $G$.
For some $s\in\mathbb{N}$, a graph $G$ is said to be $s$-degenerate if every subgraph of $G$ has minimum degree at most $s$. 
Miao, Shi, Hu, and Luo (\textit{Discrete Mathematics}, 339(10):2446--2449, 2016) showed that the AVD Total Coloring Conjecture is true for 2-degenerate graphs.  
We verify the conjecture for 3-degenerate graphs. 
\end{abstract}
\section{Introduction}

All graphs in this paper are assumed to be finite, simple, and undirected.
For any $k\in\mathbb{N}$, we let $[k]=\{1,2,\ldots,k\}$.
For a graph $G$, we denote by $V(G)$ and $E(G)$ the vertex and edge sets of $G$ respectively.
For $u \in V(G)$, we denote by $N_G(u)$ the set of neighbours of $u$ in $G$.
Further, we let $\partial_G(u)$ be the set of edges that are incident on $u$ in $G$, i.e. $\partial_G(u)=\{uv\colon v\in N_G(u)\}$. We also define $\partial_G^\star(u)=\{u\}\cup\partial_G(u)$.
The degree of a vertex $u$ in $G$, denoted by $deg_G(u)$, is defined as $deg_G(u)=|N_G(u)|$.
The maximum degree of the graph $G$, denoted by $\Delta(G)$, is defined as $\Delta(G)=\max_{u\in V(G)} deg_G(u)$.

A proper vertex coloring of a graph $G$ using colors from a set $C$ is a mapping $\phi: V(G) \rightarrow C$ such that for any $uv \in E(G)$, $\phi(u)\neq \phi(v)$.
Similarly, a proper edge coloring of $G$ using colors from a set $C$ is a mapping $\phi: E(G) \rightarrow C$ such that for any two distinct edges $e, e' \in \partial_G(u)$ for some $u \in V(G)$, $\phi(e)\neq \phi(e')$.
A total coloring $\phi$ of a graph $G$ using colors from a set $C$ is a mapping $\phi:V(G)\cup E(G)\rightarrow C$ such that $\phi|_{V(G)}$ is a proper vertex coloring of $G$, $\phi|_{E(G)}$ is a proper edge coloring of $G$, and for each $uv\in E(G)$, $\phi(u)\neq\phi(uv)$.

\begin{definition}
A total coloring $\phi$ of a graph $G$ such that for each $uv\in E(G)$, $\phi(\partial^\star(u))\neq\phi(\partial^\star(v))$ is called an \emph{Adjacent Vertex Distinguishing} (or \emph{AVD} for short) total coloring of $G$.
\end{definition}

The following conjecture is due to Zhang et al. \cite{zhang2005adjacent}.

\begin{conjecture}[AVD Total Coloring Conjecture]
Every graph $G$ has an AVD total coloring using at most $\Delta(G)+3$ colors.    
\end{conjecture}

A graph $G$ is said to be \emph{$s$-degenerate} if every subgraph of $G$ contains a vertex of degree at most $s$.  
It is immediate from the definition that every subgraph of an $s$-degenerate graph is also $s$-degenerate.
In this paper, we show that all 3-degenerate graphs satisfy the AVD Total Coloring Conjecture.

\begin{theorem}\label{thm:avdtc-3degenerate}
Every $3$-degenerate graph $G$ has an AVD total coloring using at most $\Delta(G)+3$ colors.
\end{theorem}


\section{Background}
One of the elementary results in the study of graph colorings is that every graph $G$ has a proper vertex coloring using at most $\Delta(G)+1$ colors. In fact, the classical theorem of Brooks states that at most $\Delta(G)$ colors are required unless $G$ is a complete graph or an odd cycle. As for proper edge colorings, the Vizing-Gupta Theorem states that every graph $G$ has a proper edge coloring using at most $\Delta(G)+1$ colors.
The Total Coloring Conjecture, one of the most famous open problems in the theory of graph coloring, proposed independently by Vizing~\cite{vizing1968some} and Behzad~\cite{behzad1965graphs}, states that every simple graph $G$ has a total coloring using at most $\Delta(G)+2$ colors. It is not hard to see that any total coloring of a graph $G$ requires at least $\Delta(G)+1$ colors.

Early works by Rosenfeld~\cite{rosenfeld1971total}, Vijayaditya~\cite{vijayaditya1971total}, and Kostochka~\cite{kostochka1977total,kostochka1996total} established that the Total Coloring Conjecture is true for graphs having maximum degree at most 5.
Isobe, Zhou, and Nishizeki~\cite{isobe2007total} proved that every $s$-degenerate graph $G$ having $\Delta(G)\geq 4s+3$ satisfies the Total Coloring Conjecture. In fact, they show the stronger result that any such graph $G$ has a total coloring using $\Delta(G)+1$ colors. Short proofs are known for the Total Coloring Conjecture for 3-degenerate graphs (see for example~\cite{tengthreedegen}). The conjecture has not been proven for 4-degenerate graphs in general.

A crucial difference between total coloring and AVD total coloring, which is not hard to see, is that an AVD total coloring of a supergraph of a graph $G$, when restricted to $G$ may not yield an AVD total coloring of $G$. The AVD Total Coloring Conjecture is known to be true for 4-regular graphs~\cite{papaioannou2014avdtc}, graphs $G$ having $\Delta(G)=3$~\cite{wang2007adjacent,chen2008adjacent,hulgan2009concise} and $\Delta(G)=4$~\cite{lu2017adjacent}. Hulgan~\cite{hulgan2010graph} conjectured that every graph $G$ having $\Delta(G)=3$ have an AVD total coloring using 5 colors. Note that any graph $G$ that contains two adjacent maximum degree vertices needs at least $\Delta(G)+2$ colors in any AVD total coloring. There are graphs $G$ that require $\Delta(G)+3$ colors in any AVD total coloring: odd cycles and complete graphs of odd order, for example.
The AVD Total Coloring Conjecture was shown to be true for complete graphs and bipartite graphs by Zhang et al.~\cite{zhang2005adjacent}. Huang, Wang and Yan~\cite{huang2012note} showed that every graph $G$ has an AVD total coloring using at most $2\Delta$ colors.
Hu et al.~\cite{hu2019adjacent} showed that the conjecture is true for planar of maximum degree at least 9, and Verma, Fu and Panda~\cite{verma2022adjacent} showed that the conjecture is true for split graphs.
The conjecture was shown to be true for outerplanar graphs by Wang and Wang~\cite{wang2010adjacent}. Later, Miao et al.~\cite{miao2016adjacent} showed that the conjecture is true for 2-degenerate graphs.
\section{Notation}
Let $G$ be a graph and $s \in \mathbb{N}$.
We denote by $\degleq{s}(G)$ the vertices of $G$ that have degree at most $s$. We define $\deggt{s}(G)=V(G)\setminus \degleq{s}(G)$.
For a vertex $u\in V(G)$, we define $\nbrsgt{G}{s}(u)=N_G(u)\cap\deggt{s}(G)$ and $\nbrsleq{G}{s}(u)=N_G(u)\cap\degleq{s}(G)$.
For a graph $G$ and a set $V' \subseteq V(G)$, let $G[V']$ denote the subgraph of $G$ induced by $V'$ and $G-V'$ denote the graph $G[V\setminus V']$. 
Moreover, for set $E'\subseteq E(G)$, let $G-E'$ denote the subgraph $G'$ of $G$ such that $V(G')=V(G)$ and $E(G')=E(G)\setminus E'$.

A vertex is said to be an \emph{$s$-pivot} if it has at most $s$ neighbours in $\deggt{s}(G)$, i.e. $u$ is an $s$-pivot if $|\nbrsgt{G}{s}(u)|\leq s$. Clearly, every vertex in $\degleq{s}(G)$ is by definition an $s$-pivot. The following observation is folklore.

\begin{observation}\label{obs:pivot}
Let $G$ be an $s$-degenerate graph. If $\Delta(G)> s$, then $G$ has an $s$-pivot having degree greater than $s$.
\end{observation}
\begin{proof}
Let $G'=G-\degleq{s}(G)$. Since $\Delta(G)>s$, we have that $V(G')\neq\emptyset$. As $G$ is $s$-degenerate, we have that $G'$ is also $s$-degenerate, which implies that there exists $v\in V(G')$ such that $deg_{G'}(v)\leq s$. Since $v\in V(G')$, we have that $deg_G(v)> s$ and that $\nbrsgt{G}{s}(v)= N_{G'}(v)$. As $deg_{G'}(v)\leq s$, it follows that $|\nbrsgt{G}{s}(v)|\leq s$. Thus, $v$ is an $s$-pivot in $G$ having degree greater than $s$.
\end{proof}
\section{Proof of the main result}
We shall prove the AVD Total Coloring Conjecture for 3-degenerate graphs $G$ having $\Delta(G)\geq 5$.
Any total coloring of a graph having maximum degree 1 can easily be seen to be an AVD total coloring, and clearly any such graph has a total coloring using at most 3 colors. Graphs having maximum degree 2 are disjoint unions of paths and cycles, and for these graphs, it is not hard to see that there exist AVD total colorings using at most 5 colors.
As mentioned above, the AVD Total Coloring Conjecture is known to be true for graphs $G$ having $\Delta(G)=3$ and $\Delta(G)=4$. Thus, our proof for the conjecture for 3-degenerate graphs having maximum degree at least 5 proves the conjecture for the whole class of 3-degenerate graphs.
\medskip


\noindent\textbf{Proof sketch:} Our proof will be based on induction on the number of edges of the graph.
The inductive proof yields what we call a ``partial AVD total coloring''  of $G$ (see Definition~\ref{def:partialcoloring}) using at most $\max\{8,\Delta(G)+3\}$ colors. In this kind of coloring, all the edges are colored, but some vertices in $\degleq{3}(G)$ may not be colored. All vertices in $\deggt{3}(G)$ are colored and the AVD condition is maintained between every pair of adjacent vertices in $\deggt{3}(G)$, but the AVD condition may not be satisfied for two adjacent vertices in $\degleq{3}(G)$. A partial AVD total coloring of a graph can be converted into an AVD total coloring as long as at least 7 colors are available (see Lemma~\ref{lem:conflict}).

\begin{definition}\label{def:partialcoloring}
A \emph{partial} total coloring of a graph $G$ using colors from a set $C$ is a mapping $\phi:\high(G)\cup V'\cup E(G)\rightarrow C$ for some $V'\subseteq\low(G)$ such that:
\begin{itemize}
\item $\phi|_{\high(G)\cup V'}$ is a proper vertex coloring of $G[\high(G)\cup V']$,
\item $\phi|_{E(G)}$ is a proper edge coloring of $G$, and
\item for each $uv\in E(G)$ such that $u\in \high(G)\cup V'$, $\phi(u)\neq\phi(uv)$.
\end{itemize}
If in addition to the above conditions, $\phi$ also satisfies the condition that for each $uv\in E(G[\high(G)])$, $\phi(\partial_G^\star(u))\neq\phi(\partial_G^\star(v))$, then we say that $\phi$ is a partial AVD total coloring of $G$.
\end{definition}

For the inductive step in Lemma~\ref{lem:main} to work, we need a notion that we call a ``\tbd'' for the graph. This is nothing but a list of pairs of degree 2 vertices in the graph such that the two vertices in each pair are nonadjacent, and no vertex appears in more than one pair. The formal definition is given below.

\begin{definition}
We say that a set $\mathcal{S}\subseteq\{\{u,v\}\colon u,v\in V(G), deg_G(u)=deg_G(v)=2, uv\notin E(G)\}$ is a \emph{\tbd} for a graph $G$ if we have $S\cap S'=\emptyset$ for distinct $S,S'\in\mathcal{S}$.
\end{definition}

We say that a partial total coloring of a graph $G$ ``satisfies'' a \tbd\ $\mathcal{S}$ of $G$ if for each $\{u,v\}\in \mathcal{S}$, the set of colors on the incident edges of $u$ is not the same as the set of colors on the incident edges of $v$. The formal definition is given below.
\begin{definition}
For a graph $G$ and a \tbd\ $\mathcal{S}$ of $G$, we say that a (partial) total coloring $\phi$ of $G$ \emph{satisfies} $\mathcal{S}$ if for each $\{u,v\}\in \mathcal{S}$ we have $\phi(\partial_G(u))\neq\phi(\partial_G(v))$.
\end{definition}


We are now ready to state our main result which is slightly stronger than what we require. Note that the following theorem can be applied setting $\mathcal{S}=\emptyset$ to obtain the claimed result.

\begin{theorem}\label{thm:main}
For every 3-degenerate graph $G$ having $\Delta(G)\geq 5$ and any \tbd\ $\mathcal{S}$ of $G$, there is an AVD total coloring of $G$ using at most $\Delta(G)+3$ colors that satisfies $\mathcal{S}$.
\end{theorem}

We first show that given any partial AVD total coloring of a graph $G$ using at least 7 colors, one can construct an AVD total coloring of $G$ using the same number of colors.

\begin{lemma}\label{lem:conflict}
Let $C$ be a set of colors such that $|C|\geq 7$.
For any graph $G$ and any \tbd\ $\mathcal{S}$ of $G$, if there exists a partial AVD total coloring $\phi$ of $G$ using colors from $C$ that satisfies $\mathcal{S}$, then there is an AVD total coloring of $G$ using colors from $C$ that satisfies $\mathcal{S}$.
\end{lemma}
\begin{proof}
We first observe that every partial AVD total coloring $\psi$ of $G$ using colors from $C$ in which at least one vertex is not colored can be extended to a partial AVD total coloring of $G$ using colors from $C$ with a lesser number of vertices that are not colored. This is because any vertex $u$ that is not colored in $\psi$ has degree at most 3 and therefore sees at most $|\partial_G(u)\cup N_G(u)|\leq 6$ colors on its incident edges and neighbours. As $|C|\geq 7$, there is a color in $C$ which can be assigned to $u$ to obtain the required partial AVD total coloring of $G$. Note that the new coloring satisfies a support $\mathcal{S}$ of $G$ if $\psi$ satisfies $\mathcal{S}$, since no edge was recolored in the process of constructing it from $\psi$.

Let $G$ be any graph and $\phi$ be a partial AVD total coloring of $G$ using colors from $C$ that satisfies $\mathcal{S}$. Using the procedure described in the above paragraph repeatedly, we can obtain a partial AVD total coloring $\phi'$ of $G$ using colors from $C$ that satisfies $\mathcal{S}$ in which all vertices are colored.

We say that an edge $uv\in E(G)$ is a ``violating edge'' if $\phi'(\partial_G^\star(u))=\phi'(\partial_G^\star(v))$. Clearly, since $\phi'$ is a partial AVD total coloring of $G$, if $uv$ is a violating edge, then $u,v\in\low(G)$. If there are no violating edges in $\phi'$, then $\phi'$ is an AVD total coloring of $G$, and we are done.
So let us assume that there is a violating edge $uv\in E(G)$. We shall show how to modify the color of the vertex $u$ to get a partial AVD total coloring of $G$ in which there is at least one less violating edge than in $\phi'$. Since only a vertex is getting recolored, the new coloring also satisfies $\mathcal{S}$. By repeating this procedure, we arrive at a coloring with no violating edges, and this will finish the proof.

We now show how to recolor the vertex $u$ so that the number of violating edges in the new coloring is at least one less than that in $\phi'$.
Let $X=\phi'(\partial_G^\star(v))=\phi'(\partial_G^\star(u))$.
Clearly, $|X|\leq 4$.
For each $w\in N_G(u)$ and color $c\in C\setminus X$, we say that $c$ is ``bad'' for $w$ if either $\phi'(w)=c$ or $\phi'(\partial_G^\star(w))\setminus X=\{c\}$. It immediately follows from this definition that there is no color in $C\setminus X$ that is bad for $v$. Further, it is easy to verify that for any $w\in N_G(u)$, there exists at most one color in $C\setminus X$ that is bad for $w$. Since $|C|\geq 7$, $|X|\leq 4$, and $deg_G(u)\leq 3$, we now have that there exists a color $c\in C\setminus X$ that is not bad for any neighbour of $u$.
Now we recolor $u$ with $c$ to obtain $\varphi$. We claim that $\varphi$ is a partial AVD total coloring of $G$ satisfying $\mathcal{S}$ with lesser number of violating edges than $\phi'$.

Clearly, $\varphi$ satisfies $\mathcal{S}$ as it was obtained from $\phi'$ without recoloring any edges.
Notice that the edges incident on $u$ all have colors from $X$ in $\phi'$ and therefore also in $\varphi$, and $\varphi(u)=c\notin X$.
Also, since $c$ was not bad for any neighbour of $u$, we can conclude that every neighbour of $u$ had a color different from $c$ in $\phi'$, and therefore also in $\varphi$.
Thus, $\varphi$ is a partial AVD total coloring of $G$.
Note that since $\varphi$ is obtained from $\phi'$ by just modifying the color of the vertex $u$, any edge that is not violating in $\phi'$ but is violating in $\varphi$ must be incident on $u$.
Suppose that there exists $w\in N_G(u)$ such that $uw$ is a violating edge in $\varphi$. Then $\varphi(\partial_G^\star(w))=\varphi(\partial_G^\star(u))$, which implies that $\varphi(\partial_G^\star(w))\setminus X=\varphi(\partial_G^\star(u))\setminus X=\{c\}$. But since $\phi'(\partial_G^\star(w))=\varphi(\partial_G^\star(w))$, we now have that $\phi'(\partial_G^\star(w))\setminus X=\varphi(\partial_G^\star(w))\setminus X=\{c\}$, implying that $c$ is bad for $w$. This contradicts the choice of $c$. We can therefore conclude that no edge incident on $u$ is a violating edge in $\varphi$.
Thus $\varphi$ has at least one less violating edge than $\phi'$. This completes the proof.
\end{proof}

Note that by the above lemma, it only remains to be shown that every 3-degenerate graph $G$ having $\Delta(G)\geq 5$ has a partial AVD total coloring using at most $\Delta(G)+3$ colors. First, we show that every subcubic graph has a partial AVD total coloring using at most 7 colors that satisfies any given support of it.

\begin{lemma}\label{lem:subcubic}
Let $G$ be a graph such that $\Delta(G)\leq 3$, and let $C\subseteq\mathbb{N}$ such that $|C|\geq 7$. Further, let $\mathcal{S}$ be any support of $G$. Then there is a partial AVD total coloring of $G$ using colors from $C$ that satisfies $\mathcal{S}$.
\end{lemma}
\begin{proof}
We prove this by induction on $|E(G)|$. If $|E(G)|=0$, then clearly, $\mathcal{S}=\emptyset$, and then $G$ has a trivial partial AVD total coloring using colors from $C$ that satisfies $\mathcal{S}$. So let us assume that $E(G)\neq\emptyset$.
Let $v_1v_2\in E(G)$ and let $G'=G-\{v_1v_2\}$. Let $\mathcal{S}'=\mathcal{S}\setminus\{P\in \mathcal{S}\colon\{v_1,v_2\}\cap P\neq\emptyset\}$. Since $\Delta(G')\leq 3$, we have by the inductive hypothesis that $G'$ has a partial AVD total coloring $\phi$ using colors from $C$ that satisfies $\mathcal{S}'$. Let $\phi'$ be obtained from $\phi$ by uncoloring all vertices that are colored in $\phi$.
For each $i\in\{1,2\}$, we define $v'_i\in V(G)$ as follows.
If $deg_G(v_i)=2$ and there exists $w\in V(G)$ such that $\{v_i,w\}\in \mathcal{S}$, then we define $v'_i=w$. Otherwise, we define $v'_i=v_i$.
Since $|\partial_{G'}(\{v_1,v_2\}\cup\{v'_1,v'_2\})|\leq 6$ and $|C|\geq 7$, there exists $c\in C\setminus\phi'(\partial_{G'}(\{v_1,v_2\}\cup\{v'_1,v'_2\}))$. We now construct an extension $\varphi$ of $\phi'$ by coloring the edge $v_1v_2$ with $c$. It can be verified that $\varphi$ is a partial AVD total coloring of $G$ that satisfies $\mathcal{S}$.
\end{proof}

Recall that by Observation~\ref{obs:pivot}, any 3-degenerate graph that is not subcubic contains a 3-pivot of degree at least 4. First, we deal with the case when the graph contains such a 3-pivot of having at most one neighbour of degree at most 3. Note that such a 3-pivot will have degree exactly 4.

\begin{lemma}\label{lem:smallpivot}
Let $G$ be a 3-degenerate graph, $\mathcal{S}$ be a \tbd\ of $G$, and $C\subseteq\mathbb{N}$ such that $|C|\geq 8$. Suppose that $u$ is a 3-pivot of $G$ such that $deg_G(u)=4$ and $\nbrsleq{G}{3}(u)=\{v\}$. Let $\mathcal{S}'=\mathcal{S}\setminus\{P\in \mathcal{S}\colon v\in P\}$. If there exists a partial AVD total coloring of $G'=G-\{uv\}$ using colors from $C$ that satisfies $\mathcal{S}'$, then there is a partial AVD total coloring of $G$ using colors from $C$ that satisfies $\mathcal{S}$.
\end{lemma}
\begin{proof}
Notice that $|\nbrsgt{G}{3}(u)|=3$ and that $\mathcal{S}'$ is a support of $G'$.
Suppose that there exists a partial AVD total coloring $\phi$ of $G'$ using colors from $C$ that satisfies $\mathcal{S}'$. Let $\phi'$ be obtained from $\phi$ by uncoloring all the vertices in $\low(G')$ (note that $u$ gets uncolored if it was colored in $\phi$), and then assigning a color $c\in C\setminus(\phi(\partial_{G'}(u)\cup\nbrsgt{G'}{3}(u))$ to $u$ (note that $c$ exists since $|C|\geq 8$, $|\partial_{G'}(u)|=3$ and $|\nbrsgt{G'}{3}(u)|=3$). It is easy to see that $\phi'$ is a partial AVD total coloring of $G'$ using colors from $C$ that satisfies $\mathcal{S}'$.
We define $v'\in V(G)$ as follows.
If $deg_G(v)=2$ and there exists $w\in V(G)$ such that $\{v,w\}\in \mathcal{S}$, then we define $v'=w$. Otherwise, we define $v'=v$.
We define $Y\subseteq C$ as follows.
If $v=v'$ then we let $Y=\phi'(\partial_{G'}(v))$. If $v\neq v'$, then we let $Y=\phi'(\partial_{G'}(v'))$ if $\phi'(\partial_{G'}(v))\subseteq\phi'(\partial_{G'}(v'))$, and $Y=\phi'(\partial_{G'}(v))$ otherwise. It is easy to verify that in every case, $|Y|\leq 2$ and $\phi'(\partial_{G'}(v))\subseteq Y$.

For the sake of convenience, if $\psi$ is a partial total coloring of a graph $H$ using colors from $C$ and $F\subseteq V(H)\cup E(H)$, we shall denote by $\overline{\psi}(F)$ the set $C\setminus\psi(F)$.

Let $X=\overline{\phi'}(\partial_{G'}^\star(u))$ (thus $X$ is the set of colors in $C$ that do not appear on $u$ or on the edges incident on $u$ in $\phi'$). Since $deg_{G'}(u)=3$, we have that $|X|\geq |C|-4\geq 4$. We now have that $|X\setminus Y|\geq 2$. Let $\{c_1,c_2\}\subseteq X\setminus Y$. Let $\varphi_1,\varphi_2$ be the extensions of $\phi'$ obtained by coloring the edge $uv$ with $c_1$ and $c_2$ respectively; i.e. for $i\in\{1,2\}$, $\varphi_i=\phi'\cup\{(uv,c_i)\}$. By our choice of $Y$, it follows that $\varphi_1$ and $\varphi_2$ are both partial total colorings of $G$ using colors from $C$. It can also be verified that both $\varphi_1$ and $\varphi_2$ satisfy $\mathcal{S}$. If one of them is also a partial AVD total coloring, then we are done. So we assume that neither $\varphi_1$ nor $\varphi_2$ is a partial AVD total coloring of $G$. Since $\phi'$ is a partial AVD total coloring of $G'$, this can only mean that for each $i\in\{1,2\}$, there exists $x_i\in \nbrsgt{G}{3}(u)$ such that $\varphi_i(\partial_G^\star(u))=\varphi_i(\partial_G^\star(x_i))$. Since $c_1\in\varphi_1(\partial_G^\star(u))\setminus \varphi_2(\partial_G^\star(u))$ and $c_2\in\varphi_2(\partial_G^\star(u))\setminus \varphi_1(\partial_G^\star(u))$, it follows that $x_1\neq x_2$. Notice that $deg_G(x_1)=deg_G(x_2)=deg_G(u)=4$.
As $|C|\geq 8$, we have that for each $i\in\{1,2\}$, $|\overline{\varphi_i}(\partial_G^\star(x_i))|\geq 3$.
Observe that for each $i\in\{1,2\}$, we have $\phi'(\partial_{G'}^\star(x_i))=\varphi_i(\partial_G^\star(x_i))=\varphi_i(\partial_G^\star(u))=\phi'(\partial_{G'}^\star(u))\cup\{c_i\}$, which implies that $\overline{\phi'}(\partial_{G'}^\star(x_i))=
\overline{\phi'}(\partial_{G'}^\star(u))\cap (C\setminus\{c_i\})=\overline{\phi'}(\partial_{G'}^\star(u))\setminus\{c_i\}$.
Thus, $\overline{\phi'}(\partial_{G'}^\star(x_1))\cap\overline{\phi'}(\partial_{G'}^\star(x_2))=\overline{\phi'}(\partial_{G'}^\star(u))\setminus\{c_1,c_2\}$.
Since $|\overline{\phi'}(\partial_{G'}^\star(u))|=|C|-4\geq 4$, we have that there exist distinct
$c_3,c_4\in\overline{\phi'}(\partial_{G'}^\star(x_1))\cap
\overline{\phi'}(\partial_{G'}^\star(x_2))$.
Let $\nbrsgt{G}{3}(u)\setminus\{x_1,x_2\}=\{x_3\}$.
Suppose that for some $j\in\{3,4\}$, $c_j\in\overline{\phi'}(\partial_{G'}^\star(x_3))$. Then $c_j\in\overline{\phi'}(\partial_{G'}^\star(x_1))\cap \overline{\phi'}(\partial_{G'}^\star(x_2))\cap\overline{\phi'}(\partial_{G'}^\star(x_3))$. Let $\varphi$ be obtained from $\varphi_1$ (or $\varphi_2$) by recoloring $u$ to $c_j$. Note that since it is also the case that $c_j\in\overline{\phi'}(\partial_{G'}^\star(u))\setminus\{c_1,c_2\}$, it follows that $\varphi$ is a partial total coloring of $G$ using colors from $C$ that satisfies $\mathcal{S}$.
Moreover, since for any $l\in\{1,2,3\}$, we have that $\varphi(\partial_G^\star(x_l))=\phi'(\partial_{G'}^\star(x_l))$ and $c_j\in\varphi(\partial_G^\star(u))\setminus\varphi(\partial_G^\star(x_l))$, we can conclude that $\varphi$ is a partial AVD total coloring of $G$, and we are done. So we assume that for each $j\in\{c_3,c_4\}$, we have $c_j\notin\overline{\phi'}(\partial_{G'}^\star(x_3))$, or in other words $\{c_3,c_4\}\subseteq\phi'(\partial_{G'}^\star(x_3))$. Clearly, at least one of $c_3,c_4$ is not equal to $\phi'(x_3)$. Without loss of generality, we assume that $c_3\neq\phi'(x_3)$. Then let $\varphi$ be obtained from $\varphi_1$ (or $\varphi_2$) by recoloring $u$ to $c_3$. As $c_3\notin \phi'(\partial_{G'}^\star(u))\cup\{c_1,c_2,\phi'(x_1),\phi'(x_2),\phi'(x_3)\}$, it follows that $\varphi$ is a partial total coloring of $G$ using colors from $C$ that satisfies $\mathcal{S}$ (recall that the only neighbours of $u$ that are colored in $\phi'$ are $x_1$, $x_2$, and $x_3$). Since for each $l\in\{1,2,3\}$, $\varphi(\partial_G^\star(x_l))=\phi'(\partial_{G'}^\star(x_l))$, and $\varphi(\partial_G^\star(u))\subseteq\phi'(\partial_{G'}(u))\cup\{c_1,c_2,c_3\}$, it can be verfied that $c_3\in\varphi(\partial_G^\star(u))\setminus\varphi(\partial_G^\star(x_1))$,
$c_3\in\varphi(\partial_G^\star(u))\setminus\varphi(\partial_G^\star(x_2))$, and $c_4\in\varphi(\partial_G^\star(x_3)) \setminus \varphi(\partial_G^\star(u))$. Hence we can conclude that $\varphi$ is a partial AVD total coloring of $G$.
\end{proof}

Before we prove the next lemma, which combines all the cases, we need an elementary observation.

\begin{definition}\label{def:zeta}
For two sets $A,B\subseteq\mathbb{N}$, let $\zeta(A,B)$ denote the set of all two-element subsets of $A\cup B$ that contain one element from $A$ and one element from $B$; i.e. $\zeta(A,B)=\{\{a,b\}\colon a\in A$ and $b\in B\}$.
\end{definition}

\begin{observation}\label{obs:sets}
Let $A_1,A_2\subseteq\mathbb{N}$ such that for each $i\in\{1,2\}$, we have $|A_1|,|A_2|\geq 2$ and $|A_1\cup A_2|\geq 3$.
\begin{enumerate}
    \item If $|A_1\cup A_2|=3$, then $|\zeta(A_1,A_2)|=3$.
    \item If $|A_1\cup A_2|=4$ and $|A_1|=|A_2|=2$, then $|\zeta(A_1,A_2)|=4$.
    \item In all other cases, $|\zeta(A_1,A_2)|\geq 5$.
\end{enumerate}
\end{observation}
\begin{proof}
Observe that if for distinct $x,y\in A_1\cup A_2$, we have $\{x,y\}\notin\zeta(A_1,A_2)$, then for some $i\in\{1,2\}$, we have $\{x,y\}\cap A_i=\emptyset$. This can be seen as follows. Let us assume without loss of generality that $x\in A_1$ (since $x\in A_1\cup A_2$). As $\{x,y\}\notin\zeta(A_1,A_2)$, we have that $y\notin A_2$, which implies that $y\in A_1$ (since $y\in A_1\cup A_2$). Again, since $\{x,y\}\notin\zeta(A_1,A_2)$, this implies that $x\notin A_2$. Thus, we get that $\{x,y\}\cap A_2=\emptyset$.

Suppose that $|A_1\cup A_2|=3$. Consider distinct $x,y\in A_1\cup A_2$. We claim that $\{x,y\}\in\zeta(A_1,A_2)$. Otherwise, by the observation above, we have that for some $i\in\{1,2\}$, $\{x,y\}\cap A_i=\emptyset$. This is impossible since $|A_1\cup A_2|=3$ and $|A_i|\geq2$. Thus, for any distinct $x,y\in A_1\cup A_2$, we have $\{x,y\}\in\zeta(A_1,A_2)$. Thus $|\zeta(A_1,A_2)|={3\choose 2}=3$.

Next, suppose that $|A_1\cup A_2|=4$, $|A_1|=2$, and $|A_2|=2$. Clearly, we then have $A_1\cap A_2=\emptyset$, and therefore it easily follows that $|\zeta(A_1,A_2)|=4$.

So we assume without loss of generality that $|A_1\cup A_2|\geq 4$ and $|A_1|\geq 3$. If for each pair of distinct $x,y\in A_1\cup A_2$, we have $\{x,y\}\in\zeta(A_1,A_2)$, then we have $|\zeta(A_1,A_2)|\geq {4\choose 2}=6$, and we are done. So we assume that there exist distinct $x,y\in A_1\cup A_2$ such that $\{x,y\}\notin\zeta(A_1,A_2)$. Then by the observation above, we know that there exists $i\in\{1,2\}$ such that $\{x,y\}\cap A_i=\emptyset$. Suppose first that $\{x,y\}\cap A_1=\emptyset$.
Then $x,y\in A_2$ and for each $z\in A_1$, we have that $\{z,x\},\{z,y\}\in\zeta(A_1,A_2)$. Since $|A_1|\geq 3$, this implies that $|\zeta(A_1,A_2)|\geq 6$, and so we are done. We finally consider the case when $\{x,y\}\cap A_2=\emptyset$. This means that $x,y\in A_1$. Since $|A_2|\geq 2$, there exist distinct $z_1,z_2\in A_2$. Clearly, $\{z_1,x\},\{z_1,y\},\{z_2,x\},\{z_2,y\}\in\zeta(A_1,A_2)$. As $|A_1|\geq 3$, there exists some $w\in A_1\setminus\{x,y\}$. Clearly, at least one of $z_1$ or $z_2$, say $z_1$, is different from $w$. Then $\{z_1,w\}\in\zeta(A_1,A_2)$. It is easy to verify that the sets $\{z_1,x\},\{z_1,y\},\{z_2,x\},\{z_2,y\},\{z_1,w\}$ are all pairwise distinct. Thus $|\zeta(A_1,A_2)|\geq 5$.
\end{proof}
\begin{lemma}\label{lem:main}
Let $k\geq 5$, $C=[k+3]$, $G$ be a 3-degenerate graph having $\Delta(G)\leq k$, and $\mathcal{S}$ be any \tbd\ of $G$. Then there is a partial AVD total coloring of $G$ using colors from $C$ that satisfies $\mathcal{S}$.
\end{lemma}
\begin{proof}
We use induction on $|E(G)|$. If $\Delta(G)\leq 3$, then we are done by Lemma~\ref{lem:subcubic}. So let us assume that $\Delta(G)\geq 4$.
Then we have by Observation~\ref{obs:pivot} that there exists a 3-pivot $u$ in $G$ such that $deg_G(u)\geq 4$. Since $u$ is a 3-pivot, it has at least one neighbour in $\low(G)$. Suppose that $u$ has exactly one neighbour in $\low(G)$. Then $deg_G(u)=4$. Let $\nbrsleq{G}{3}(u)=\{v\}$, $G'=G-\{uv\}$, and $\mathcal{S}'=\mathcal{S}\setminus\{P\in \mathcal{S}\colon v\in P\}$ (clearly, $\mathcal{S}'$ is a support of $G'$). By the inductive hypothesis, we have that there exists a partial AVD total coloring of $G'$ using colors from $C$ that satisfies $\mathcal{S}'$. Then we are done by Lemma~\ref{lem:smallpivot}.

We now consider the case when $u$ has at least two neighbours in $\low(G)$.
Let $v_1,v_2\in \nbrsleq{G}{3}(u)$.
For each $i\in\{1,2\}$, we define $v'_i\in V(G)$ as follows.
If $deg_G(v_i)=2$ and there exists $w\in V(G)$ such that $\{v_i,w\}\in \mathcal{S}$, then we define $v'_i=w$. Otherwise, we define $v'_i=v_i$.
Note that $v'_1,v'_2,u$ are all pairwise distinct.

Let $G'=G-\{uv_1,uv_2\}$.
Let $\mathcal{S}_0=\mathcal{S}\setminus\{P\in \mathcal{S}\colon\{v_1,v_2\}\cap P\neq\emptyset\}$.
If $deg_{G'}(v'_1)=deg_{G'}(v'_2)=2$ and $v'_1v'_2\notin E(G')$, then define $\mathcal{S}'=\mathcal{S}_0\cup\{\{v'_1,v'_2\}\}$. Otherwise, we let $\mathcal{S}'=\mathcal{S}_0$. It is not difficult to verify that in each case, $\mathcal{S}'$ is a pairwise disjoint collection of pairs of nonadjacent degree 2 vertices in $G'$; so $\mathcal{S}'$ is a support of $G'$.

By the inductive hypothesis, there exists a
partial AVD total coloring $\phi$ of $G'$ using colors from $C$ that satisfies $\mathcal{S}'$.
Let $\phi'$ be the partial AVD total coloring of $G'$ obtained from $\phi$ by uncoloring all the colored vertices in $\low(G')$. Clearly, $\phi'$ also satisfies $\mathcal{S}'$.
If $u\in\low(G')$, then choose $c\in C\setminus(\phi'(\partial_{G'}(u)\cup\nbrsgt{G'}{3}(u))$ (note that $c$ exists since $|C|\geq 8$, $|\partial_{G'}(u)|\leq 3$ and $|\nbrsgt{G'}{3}(u)|\leq 3$) and let $\phi''=\phi'\cup\{(u,c)\}$, otherwise let $\phi''=\phi'$. Note that $\phi''$ is a partial AVD total coloring of $G'$ that satisfies $\mathcal{S}'$ in which no vertex in $\low(G')\setminus\{u\}$ is colored.

As before, for a partial total coloring $\psi$ of a graph $H$ using colors from $C$ and $F\subseteq V(H)\cup E(H)$, we let $\overline{\psi}(F)=C\setminus\psi(F)$.

Let $X=\overline{\phi''}(\partial_{G'}^\star(u))$. Since $deg_{G'}(u)\leq\Delta(G)-2$, we have that $|\partial_{G'}^\star(u)|\leq\Delta(G)-1\leq k-1$, and therefore $|X|\geq |C|-(k-1)\geq 4$.

We now define another partial AVD total coloring $\varphi$ of $G'$ satisfying $\mathcal{S}'$, which shall be used for the remainder of the proof.
If $|X|>4$, or $deg_{G'}(v'_1)\neq 2$, or $deg_{G'}(v'_2)\neq 2$, or $v'_1v'_2\notin E(G)$, then we define $\varphi=\phi''$.
Otherwise, we have $|X|=4$, $deg_{G'}(v'_1)=deg_{G'}(v'_2)=2$ and $v'_1v'_2\in E(G')$. If $\phi''(v'_1v'_2)\notin X$, then we again define $\varphi=\phi''$. Otherwise, we have $\phi''(v'_1v'_2)\in X$. Then choose a color $c'\in C\setminus (X\cup\phi''(\partial_{G'}(\{v'_1,v'_2\})))$ (note that $c'$ exists since $|C|\geq 8$ and $|X|=4$) and assign it to $v'_1v'_2$. Define $\varphi$ to be the coloring so obtained. It is easy to verify that in any case, $\varphi$ is also a partial AVD total coloring of $G'$. Suppose that $\varphi$ does not satisfy $\mathcal{S}'$. Then it must be the case that there exists $P\in\mathcal{S}'$ such that either $v'_1\in P$ or $v'_2\in P$. We assume without loss of generality that $v'_1\in P$. Since $v'_1v'_2\in E(G')$, we know that $\mathcal{S}'=\mathcal{S}_0\subseteq\mathcal{S}$. We have from the definition of $v'_1$ that $P=\{v_1,v'_1\}$. This contradicts the fact that $P\in\mathcal{S}'=\mathcal{S}_0$. We can thus conclude that $\varphi$ satisfies $\mathcal{S}'$. Note that $\varphi$ also satisfies the following property:
\begin{quote}
    If $|X|=4$, $deg_{G'}(v'_1)=deg_{G'}(v'_2)=2$ and $v'_1v'_2\in E(G')$, then $\varphi(v'_1v'_2)\notin X$.\hfill ($*$)
\end{quote}    
Observe that $X=\overline{\phi''}(\partial_{G'}^\star(u))=\overline{\varphi}(\partial_{G'}^\star(u))$.

Two extensions $\varphi_1,\varphi_2$ of $\varphi$ which are partial total colorings of $G$ are said to be ``equivalent'' if $\varphi_1(\{uv_1,uv_2\})=\varphi_2(\{uv_1,uv_2\})$. Otherwise, we say that $\varphi_1$ and $\varphi_2$ are ``non-equivalent''.

\begin{claim}
Suppose that $\varphi$ has four pairwise non-equivalent extensions which are all partial total colorings of $G$ using colors from $C$, then at least one of them is a partial AVD total coloring of $G$.
\end{claim}

Let $\varphi_1,\varphi_2,\varphi_3,\varphi_4$ be four pairwise non-equivalent extensions of $\varphi$ which are partial total colorings of $G$ using colors from $C$.
Clearly, for each $i\in\{1,2,3,4\}$, since $\varphi_i$ is a partial total coloring of $G$ that is an extension of $\varphi$, we have that $\{\varphi_i(uv_1),\varphi_i(uv_2)\}\subseteq X$. Let $Y_i=X\setminus\{\varphi_i(uv_1),\varphi_i(uv_2)\}$. Notice that $Y_i=\overline{\varphi_i}(\partial_G^\star(u))$. For distinct $i,j\in\{1,2,3,4\}$, since $\varphi_i$ and $\varphi_j$ are non-equivalent, we have that $\varphi_i(\{uv_1,uv_2\})\neq \varphi_j(\{uv_1,uv_2\})$, which implies that $Y_i\neq Y_j$. As $u$ is a 3-pivot, we have that $|\nbrsgt{G}{3}(u)|\leq 3$.
Observe that since $v_1,v_2\in\low(G)$, neither of the edges $uv_1,uv_2$ can be incident on a vertex $x\in\nbrsgt{G}{3}(u)$, and therefore for every $x\in \nbrsgt{G}{3}(u)$ and every $i\in\{1,2,3,4\}$, $\overline{\varphi_i}(\partial_{G'}^\star(x))=\overline{\varphi}(\partial_G^\star(x))$.
Since $Y_1,Y_2,Y_3,Y_4$ are pairwise distinct, it follows that there exists some $l\in\{1,2,3,4\}$ such that for every $x\in \nbrsgt{G}{3}(u)$, $Y_{l}\neq\overline{\varphi}(\partial_{G'}^\star(x))$. 
We now have that for any $x\in \nbrsgt{G}{3}(u)$, $\overline{\varphi_{l}}(\partial_G^\star(u))=Y_{l}\neq\overline{\varphi}(\partial_{G'}^\star(x))=\overline{\varphi_{l}}(\partial_G^\star(x))$.
It follows that $\varphi_{l}$ is a partial AVD total coloring of $G$. This proves the claim.
\medskip

\begin{claim}
Suppose that there exist three pairwise non-equivalent extensions $\varphi_1,\varphi_2,\varphi_3$ of $\varphi$ using colors from $C$, each of which is a partial total coloring of $G$ satisfying $\mathcal{S}$, such that $|\bigcup_{i\in\{1,2,3\}} \varphi_i(\{uv_1,uv_2\})|<|X|$, then there exists a partial AVD total coloring of $G$ using colors from $C$ satisfying $\mathcal{S}$.
\end{claim}
There is nothing to prove if one of $\varphi_1,\varphi_2,\varphi_3$ is also AVD. So we assume that none of them are AVD. Since $\varphi$ is a partial AVD total coloring of $G'$ and $\varphi_1,\varphi_2,\varphi_3$ are extensions of $\varphi$, it follows that for each $i\in\{1,2,3\}$, there exists $x_i\in \nbrsgt{G}{3}(u)$ such that $\varphi_i(\partial_G^\star(u))=\varphi_i(\partial_G^\star(x_i))$. Since $\varphi_1,\varphi_2,\varphi_3$ are pairwise non-equivalent extensions of $\varphi$, we have that $\varphi_1(\partial_G^\star(u)), \varphi_2(\partial_G^\star(u)), \varphi_3(\partial_G^\star(u))$ are pairwise distinct. This implies that $x_1,x_2,x_3$ are pairwise distinct, from which it follows that $\nbrsgt{G}{3}(u)=\{x_1,x_2,x_3\}$.
Choose $c\in X\setminus \bigcup_{i\in\{1,2,3\}} \varphi_i(\{uv_1,uv_2\})$ (note that $c$ exists by the assumption in the statement of the claim).
Now let $\varphi'$ be obtained from $\varphi_1$ (or $\varphi_2$ or $\varphi_3$) by recoloring $u$ with $c$.
For each $i\in\{1,2,3\}$, since $c\in X\setminus\varphi_i(\{uv_1,uv_2\})$, we have that $c\notin\varphi_i(\partial_G^\star(u))=\varphi_i(\partial_G^\star(x_i))
=\varphi(\partial_{G'}^\star(x_i))=\varphi'(\partial_G^\star(x_i))$ (the second equality is due to the fact that the edges $uv_1$ and $uv_2$ cannot be incident on $x_i$). In particular, $c\notin\varphi'(\{x_1,x_2,x_3\})$.
Also, since $c\notin\varphi_i(\partial_G^\star(u))$ for any $i\in\{1,2,3\}$, we have that $c\notin\varphi'(\partial_G(u))$. It follows that $\varphi'$ is a partial total coloring of $G$ satisfying $\mathcal{S}$.
Also, we have $c\in\varphi'(\partial_G^\star(u))\setminus\varphi'(\partial_G^\star(x_i))$ for each $i\in\{1,2,3\}$, which means that $\varphi'(\partial_G^\star(u))$ is different from $\varphi'(\partial_G^\star(x))$ for each $x\in \nbrsgt{G}{3}(u)$. Thus $\varphi'$ is a partial AVD total coloring of $G$ using colors from $C$. It is easy to see that $\varphi'$ satisfies $\mathcal{S}$. This proves the claim.
\medskip

Suppose that $\{v_1,v_2\}\in \mathcal{S}$.
Then $deg_{G'}(v_1)=deg_{G'}(v_2)=1$ and $v_1v_2\notin E(G')$.
Recall that in this case, we have $\mathcal{S}'=\mathcal{S}\setminus\{\{v_1,v_2\}\}$.
Let $\partial_{G'}(v_1)=\{e_1\}$ and $\partial_{G'}(v_2)=\{e_2\}$.
Further let $\varphi(e_1)=c_1$ and $\varphi(e_2)=c_2$.
Let $\mathcal{F}$ be the set of extensions of $\varphi$ using colors from $C$ which are partial total colorings of $G$.
Further, let $\mathcal{F}'\subseteq\mathcal{F}$ be a maximal set of pairwise non-equivalent extensions of $\varphi$ in $\mathcal{F}$. 

First suppose that $c_1=c_2=c$ (say).
Observe that any extension of $\varphi$ is a partial total coloring of $G$ if and only if it assigns two distinct colors from $X\setminus\{c\}$ to $uv_1$ and $uv_2$. In particular, for any $Z\in\zeta(X\setminus\{c\},X\setminus\{c\})$, there exists an extension $\psi\in\mathcal{F}$ such that $\psi(\{uv_1,uv_2\})=Z$ (recall Definition~\ref{def:zeta}). As $|X|\geq 4$, we have $|X\setminus\{c\}|\geq 3$. Then clearly, we have $\zeta(X\setminus\{c\},X\setminus\{c\})\neq\emptyset$, and therefore $\mathcal{F}\neq\emptyset$. Since $\mathcal{F}'$ is a maximal set of pairwise non-equivalent extensions in $\mathcal{F}$, we have 
$\{\psi(\{uv_1,uv_2\})\colon\psi\in\mathcal{F}'\}=\zeta(X\setminus\{c\},X\setminus\{c\})$.
It follows that $|\mathcal{F}'|=|\{\psi(\{uv_1,uv_2\})\colon\psi\in\mathcal{F}'\}|=|\zeta(X\setminus\{c\},X\setminus\{c\})|$. Note that any extension $\psi\in\mathcal{F}'$ also satisfies $\mathcal{S}$, since $\varphi$ satisfies $\mathcal{S}'$, $\psi(uv_1)\neq\psi(uv_2)$, and $\psi(e_1)=c=\psi(e_2)$.
Recall that $|X\setminus\{c\}|\geq 3$.
If $|\mathcal{F}'|\geq 4$, then we have by Claim 1 that there is an extension in $\mathcal{F'}$ that is a partial AVD total coloring of $G$ that satisfies $\mathcal{S}$, and we are done.
Otherwise, $|\zeta(X\setminus\{c\},X\setminus\{c\})|=|\mathcal{F}'|\leq 3$, which implies by Observation~\ref{obs:sets} that $|X\setminus\{c\}|=3$ and $|\mathcal{F}'|=|\zeta(X\setminus\{c\},X\setminus\{c\})|=3$.
(This implies that $c\in X$ and $|X|=4$.)
Now $|\bigcup_{\psi\in\mathcal{F}'} \psi(\{uv_1,uv_2\})|\leq |X\setminus\{c\}|<|X|$. 
Then we have by Claim~2 that there exists a partial AVD total coloring of $G$ using colors from $C$ that satisfies $\mathcal{S}$, and we are done.

So we assume that $c_1 \neq c_2$. 
It is not hard to see that any extension of $\varphi$ that is a partial total coloring of $G$ has to assign a color from $X\setminus\{c_1\}$ to $uv_1$ and a different color from $X\setminus\{c_2\}$ to $uv_2$. As before, it follows from the fact that $\mathcal{F}'$ is a maximal set of pairwise non-equivalent extensions in $\mathcal{F}$ that $|\mathcal{F}'|=|\zeta(X\setminus\{c_1\},X\setminus\{c_2\})|$.
Clearly, $|X\setminus\{c_i\}| \geq 3$ for each $i\in\{1,2\}$. We consider two cases now. In the first case, if $X\setminus\{c_1\}=X\setminus\{c_2\}$, then $c_1,c_2\notin X$, which implies that $|X\setminus\{c_1\}|=|X\setminus\{c_2\}|=|X|\geq 4$. In the second case, we have $X\setminus\{c_1\}\neq X\setminus\{c_2\}$. From Observation~\ref{obs:sets}, it can be  seen that in both cases, $|\zeta(X\setminus\{c_1\},X\setminus\{c_2\})|\geq 5$, which implies that $|\mathcal{F}'|\geq 5$.
Note that since $\varphi$ satisfies $\mathcal{S}'$, it follows that if an extension $\psi\in\mathcal{F}$
does not satisfy $\mathcal{S}$, then $\psi(uv_1)=c_2$ and $\psi(uv_2)=c_1$. Since the extensions in $\mathcal{F}'$ are pairwise non-equivalent, all but at most one extension in $\mathcal{F}'$ satisfies $\mathcal{S}$. Thus we have that there are four extensions in $\mathcal{F}'$ that satisfy $\mathcal{S}$.
Hence, one of them is a partial AVD total coloring of $G$ that satisfies $\mathcal{S}$ by Claim 1, and we are done. 
\medskip

We shall now assume that $\{v_1,v_2\}\notin \mathcal{S}$. Note that this implies that $v'_1\neq v_2$ and $v'_2\neq v_1$.
We now define two sets $X_1,X_2\subseteq X$ and then assign a color from $X_1$ to $uv_1$ and a color from $X_2$ to $uv_2$ to obtain a partial total coloring of $G$.
For each $i\in\{1,2\}$, we define $X_i$ as follows.
If $v_i=v'_i$, then we let $X_i=X\setminus\varphi(\partial_{G'}(v_i))$.
On the other hand, if $v_i\neq v'_i$, then we define $X_i=X\setminus\varphi(\partial_{G'}(v'_i))$ if $\varphi(\partial_{G'}(v_i))\subseteq\varphi(\partial_{G'}(v'_i))$, and $X_i=X\setminus\varphi(\partial_{G'}(v_i))$ otherwise. Note that we always have $X_i\cap\varphi(\partial_{G'}(v_i))=\emptyset$.
It is easy to verify that any extension of $\varphi$ that assigns a color from $X_1$ to $uv_1$ and a different color from $X_2$ to $uv_2$ is a partial total coloring of $G$. In other words, any $\psi=\varphi\cup\{(uv_1,c_1),(uv_2,c_2)\}$, where $c_1\in X_1$ and $c_2\in X_2\setminus\{c_1\}$ is a partial total coloring of $G$. Let $\mathcal{F}$ denote the set of all such extensions of $\varphi$.

We claim that any $\psi\in\mathcal{F}$ satisfies $\mathcal{S}$. Suppose for the sake of contradiction that $\psi=\varphi\cup\{(uv_1,c_1),(uv_2,c_2)\}$, where $c_1\in X_1$ and $c_2\in X_2\setminus\{c_1\}$, does not satisfy $\mathcal{S}$.
Then there exists $\{w,z\}\in \mathcal{S}$ such that $\psi(\partial_G(w))=\psi(\partial_G(z))$. Suppose that $\{w,z\}\in \mathcal{S}'$. Then we have $deg_G(w)=deg_G(z)=deg_{G'}(w)=deg_{G'}(z)=2$. Since $w$ and $z$ have the same degree in both $G$ and $G'$, it is clear that neither of the edges $uv_1$ or $uv_2$ are incident on $w$ or $z$ in $G$. Thus $\partial_G(w)=\partial_{G'}(w)$ and $\partial_G(z)=\partial_{G'}(z)$. Since $\psi$ is an extension of $\varphi$, we now have that $\varphi(\partial_{G'}(w))=\psi(\partial_G(w))=\psi(\partial_G(z))=\varphi(\partial_{G'}(z))$. 
This contradicts the fact that $\varphi$ satisfies $\mathcal{S}'$. So we can conclude that $\{w,z\}\notin \mathcal{S}'$, or in other words, $\{w,z\}\in \mathcal{S}\setminus \mathcal{S}'$. From the construction of $\mathcal{S}'$, we have that $\{w,z\}\in \mathcal{S}\setminus \mathcal{S}_0$, or in other words $\{w,z\}\cap\{v_1,v_2\}\neq\emptyset$. 
Let us assume without loss of generality that $w=v_1$. Then it follows that $v'_1=z$ (and also that $v_1\neq v'_1$). Note that we now have that $\psi(\partial_G(v_1))=\psi(\partial_G(v'_1))$. Since we have assumed that $v'_1\neq v_2$, we have that $\partial_G(v'_1)=\partial_{G'}(v'_1)$, which implies that $\psi(\partial_G(v'_1))=\varphi(\partial_{G'}(v'_1))$. Note that $deg_{G'}(v_1)=1$ since $\{v_1,v'_1\}\in \mathcal{S}$. As $\psi$ is an extension of $\varphi$, we now have $\varphi(\partial_{G'}(v_1))\subseteq\psi(\partial_G(v_1))=\psi(\partial_G(v'_1))=\varphi(\partial_{G'}(v'_1))$.
It now follows from the definition of $X_1$ that $X_1=X\setminus\varphi(\partial_{G'}(v'_1))$, which implies that $c_1=\psi(uv_1)\notin\varphi(\partial_{G'}(v'_1))=\psi(\partial_G(v'_1))$. Thus we have that $c_1\in\psi(\partial_G(v_1))\setminus\psi(\partial_G(v'_1))$. This contradicts the fact that $\psi(\partial_G(v_1))=\psi(\partial_G(v'_1))$.
This shows that any $\psi\in\mathcal{F}$ is a partial total coloring of $G$ that satisfies $\mathcal{S}$.

It is not hard to verify that for each $i\in\{1,2\}$, we have $deg_{G'}(v_i)\leq 2$ and $deg_{G'}(v'_i)\leq 2$. It follows that $|X_i|\geq 2$ for each $i\in\{1,2\}$.
We claim that $|X_1\cup X_2|\geq 3$. Suppose for the sake of contradiction that $|X_1\cup X_2|=2$. Then $|X_1|=|X_2|=2$ and $X_1=X_2$. Note that the former implies by the definition of $X_1$ and $X_2$ that $|X|=4$.
Let $i\in\{1,2\}$.
From the definition of $X_i$, we can see that either $X_i=X\setminus\varphi(\partial_{G'}(v_i))$ or $X_i=X\setminus\varphi(\partial_{G'}(v'_i))$.
Suppose that $X_i=X\setminus\varphi(\partial_{G'}(v_i))$. Then since $|X_i|=2$, we have that $deg_{G'}(v_i)=2$, which implies that $deg_G(v_i)=3$, and therefore $v'_i=v_i$. Thus, we have that in any case, $X_i=X\setminus\varphi(\partial_{G'}(v'_i))$.
Since $X_1=X_2$, we now have that $\varphi(\partial_{G'}(v'_1))=\varphi(\partial_{G'}(v'_2))$. Since $\varphi$ satisfies $\mathcal{S}'$, it follows that $\{v'_1,v'_2\}\notin \mathcal{S}'$. Note that as $|X_1|=|X_2|=2$ and $|X|=4$, we have that $deg_{G'}(v'_1)=deg_{G'}(v'_2)=2$. It now follows from the construction of $\mathcal{S}'$ that $v'_1v'_2\in E(G')$.
From ($*$),
we now have that $\varphi(v'_1v'_2)\notin X$. Then we have that for each $i\in\{1,2\}$, $|\varphi(\partial_{G'}(v'_i))\cap X|\leq 1$, which implies that $|X_i|=|X\setminus\varphi(\partial_{G'}(v'_i))|\geq 3$, which is a contradiction. This proves that $|X_1\cup X_2|\geq 3$.

Let $\mathcal{F}'\subseteq\mathcal{F}$ be a maximal set of pairwise non-equivalent extensions of $\varphi$ from $\mathcal{F}$.
It is not hard to verify that $|\mathcal{F}'|=|\zeta(X_1,X_2)|$. Suppose that $|X_1\cup X_2|=3$. Then we have from Observation~\ref{obs:sets} that $|\mathcal{F}'|=|\zeta(X_1,X_2)|=3$. Let $\mathcal{F}'=\{\psi_1,\psi_2,\psi_3\}$. Notice that for each $i\in\{1,2,3\}$, we have $\psi_i(\{uv_1,uv_2\})\subseteq X_1\cup X_2$.
Since $|X_1\cup X_2|=3$, we have that $|\bigcup_{i\in\{1,2,3\}} \psi_i(\{uv_1,uv_2\})|\leq |X_1\cup X_2|<|X|$. We now have from Claim~2 that there exists a partial AVD total coloring of $G$ using colors from $C$ that satisfies $\mathcal{S}$, and so we are done. Hence we assume that $|X_1\cup X_2|>3$. We now have from Observation~\ref{obs:sets} that $|\mathcal{F}'|=|\zeta(X_1,X_2)|\geq 4$. Now we have from Claim~1 that at least one extension of $\varphi$ in $\mathcal{F}'$ is a partial AVD total coloring of $G$. This extension, being from $\mathcal{F}$, also satisfies $\mathcal{S}$, thereby completing the proof.
\end{proof}

\subsubsection*{Proof of Theorem~\ref{thm:main}}

Let $k=\Delta(G)$ and $C=[k+3]$. Since $\Delta(G)\geq 5$, we have $k\geq 5$ and $|C|\geq 8$.
By Lemma~\ref{lem:main}, we now have that there exists a partial AVD total coloring $\phi$ of $G$ using colors from $C$ that satisfies $\mathcal{S}$. Applying Lemma~\ref{lem:conflict}, we get that there exists an AVD total coloring of $G$ using colors from $C$ that satisfies $\mathcal{S}$. Since $|C|=\Delta(G)+3$, we are done.\hfill\qed
\bibliographystyle{plain}
\bibliography{reference.bib}

\begin{thebibliography}{10}

\bibitem{behzad1965graphs}
Mehdi Behzad.
\newblock {\em Graphs and {T}heir {C}hromatic {N}umbers}.
\newblock Michigan State University, 1965.

\bibitem{chen2008adjacent}
Xiang’en Chen.
\newblock On the adjacent vertex distinguishing total coloring numbers of graphs with {$\Delta$}= 3.
\newblock {\em Discrete Mathematics}, 308(17):4003--4007, 2008.

\bibitem{hu2019adjacent}
Jie Hu, Guanghui Wang, Jianliang Wu, Donglei Yang, and Xiaowei Yu.
\newblock Adjacent vertex distinguishing total coloring of planar graphs with maximum degree 9.
\newblock {\em Discrete Mathematics}, 342(5):1392--1402, 2019.

\bibitem{huang2012note}
Danjun Huang, Weifan Wang, and Chengchao Yan.
\newblock A note on the adjacent vertex distinguishing total chromatic number of graphs.
\newblock {\em Discrete Mathematics}, 312(24):3544--3546, 2012.

\bibitem{hulgan2009concise}
Jonathan Hulgan.
\newblock Concise proofs for adjacent vertex-distinguishing total colorings.
\newblock {\em Discrete Mathematics}, 309(8):2548--2550, 2009.

\bibitem{hulgan2010graph}
Jonathan Hulgan.
\newblock {\em Graph colorings with constraints ({PhD} Thesis)}.
\newblock The University of Memphis, 2010.

\bibitem{isobe2007total}
Shuji Isobe, Xiao Zhou, and Takao Nishizeki.
\newblock Total colorings of degenerate graphs.
\newblock {\em Combinatorica}, 27(2):167--182, 2007.

\bibitem{kostochka1977total}
Alexandr~V Kostochka.
\newblock The total coloring of a multigraph with maximal degree 4.
\newblock {\em Discrete Mathematics}, 17(2):161--163, 1977.

\bibitem{kostochka1996total}
Alexandr~V Kostochka.
\newblock The total chromatic number of any multigraph with maximum degree five is at most seven.
\newblock {\em Discrete Mathematics}, 162(1-3):199--214, 1996.

\bibitem{lu2017adjacent}
You Lu, Jiaao Li, Rong Luo, and Zhengke Miao.
\newblock Adjacent vertex distinguishing total coloring of graphs with maximum degree 4.
\newblock {\em Discrete Mathematics}, 340(2):119--123, 2017.

\bibitem{miao2016adjacent}
Zhengke Miao, Rui Shi, Xiaolan Hu, and Rong Luo.
\newblock Adjacent vertex distinguishing total colorings of 2-degenerate graphs.
\newblock {\em Discrete Mathematics}, 339(10):2446--2449, 2016.

\bibitem{papaioannou2014avdtc}
A~Papaioannou and C~Raftopoulou.
\newblock On the {AVDTC} of 4-regular graphs.
\newblock {\em Discrete Mathematics}, 330:20--40, 2014.

\bibitem{rosenfeld1971total}
M~Rosenfeld.
\newblock On the total coloring of certain graphs.
\newblock {\em Israel Journal of Mathematics}, 9(3):396--402, 1971.

\bibitem{tengthreedegen}
Yang Teng-fei and XU~Chang-qing.
\newblock Total colorings of 3-degenerate graphs.
\newblock {\em Journal of Shandong University (Natural Science)}, 57(6):61--63, 2022.

\bibitem{verma2022adjacent}
Shaily Verma, Hung-Lin Fu, and BS~Panda.
\newblock Adjacent vertex distinguishing total coloring in split graphs.
\newblock {\em Discrete Mathematics}, 345(11):113061, 2022.

\bibitem{vijayaditya1971total}
N~Vijayaditya.
\newblock On total chromatic number of a graph.
\newblock {\em Journal of the London Mathematical Society}, 2(3):405--408, 1971.

\bibitem{vizing1968some}
Vadim~G Vizing.
\newblock Some {U}nsolved {P}roblems in {G}raph {T}heory.
\newblock {\em Russian Mathematical Surveys}, 23(6):125, 1968.

\bibitem{wang2007adjacent}
Haiying Wang.
\newblock On the adjacent vertex-distinguishing total chromatic numbers of the graphs with {$\Delta$(G)}= 3.
\newblock {\em Journal of combinatorial optimization}, 14(1):87--109, 2007.

\bibitem{wang2010adjacent}
Yiqiao Wang and Weifan Wang.
\newblock Adjacent vertex distinguishing total colorings of outerplanar graphs.
\newblock {\em Journal of combinatorial optimization}, 19(2):123--133, 2010.

\bibitem{zhang2005adjacent}
Zhongfu Zhang, Xiang’en Chen, Jingwen Li, Bing Yao, Xinzhong Lu, and Jianfang Wang.
\newblock On adjacent-vertex-distinguishing total coloring of graphs.
\newblock {\em Science in China Series A: Mathematics}, 48(3):289--299, 2005.

\end{thebibliography}

\end{document}